\documentclass{llncs}

\usepackage[T1]{fontenc}
\usepackage[dvipsnames]{xcolor}
\usepackage{amssymb,amsmath}
\usepackage[disable]{todonotes}
\usepackage[linesnumbered,ruled,noend]{algorithm2e}

\usepackage{tikz,pgffor}
\usetikzlibrary{arrows}
\usetikzlibrary{shapes}
\usetikzlibrary{calc}
\usetikzlibrary{automata}
\tikzstyle{rep}=[rounded rectangle,thick,draw,minimum size=1.4em,inner sep=1ex]
\tikzstyle{op}=[rounded rectangle,draw,minimum size=1.4em,inner sep=1ex]
\tikzstyle{fail}=[rectangle,thick,draw,minimum size=1.4em,inner sep=1ex]
\tikzstyle{tran}=[draw,->,>=stealth]

\newcommand{\A}{\mathcal{A}}
\title{Synthesis of Optimal Resilient Control Strategies\thanks{
	The authors are partly supported by the Czech Science Foundation, grant No.~15-17564S,
	by the DFG through the Collaborative Research Center SFB 912 -- HAEC,
	the Excellence Initiative by the German Federal and State Governments 
	(cluster of excellence cfAED), and the DFG-projects BA-1679/11-1 and BA-1679/12-1.}
}%

\author{
	Christel Baier\inst{1}
\and
	Clemens Dubslaff\inst{1}
\and
    \v{L}ubo\v{s} Koren\v{c}iak\inst{2}
\and
   	Anton\'{\i}n Ku\v{c}era\inst{2}
\and
	Vojt\v{e}ch \v{R}eh\'ak\inst{2}
}

\institute{
  TU Dresden, Germany\\
  \email{\{christel.baier, clemens.dubslaff\}@tu-dresden.de}
\and
   Masaryk University, Brno, Czech Republic\\
   \email{\{korenciak, kucera, rehak\}@fi.muni.cz}
}

\newcommand{\cB}{\mathcal{B}}

\newcommand{\cE}{\mathcal{E}}

\newcommand{\cL}{\mathcal{L}}
\newcommand{\cM}{\mathcal{M}}
\newcommand{\cN}{\mathcal{N}}

\newcommand{\cP}{\mathcal{P}}
\newcommand{\cQ}{\mathcal{Q}}

\newcommand{\dist}{\mathrm{Dist}}
\newcommand{\Act}{\mathit{Act}}
\newcommand{\act}{\alpha}
\newcommand{\sinit}{s_{\mathit{\scriptscriptstyle init}}}

\newcommand{\E}{\Exp}
\newcommand{\Exp}{\ensuremath{\mathbb{E}}}

\newcommand{\wgt}{\mathit{wgt}}
\newcommand{\rew}{\mathit{rew}}
\newcommand{\cost}{\mathit{cost}}
\newcommand{\payoff}{\mathit{payoff}}

\newcommand{\fMP}{\mathrm{MP}}

\newcommand{\sched}{\mathfrak{S}}
\newcommand{\asched}{\mathfrak{H}}
\newcommand{\tsched}{\mathfrak{T}}
\newcommand{\rsched}{\mathfrak{R}}
\newcommand{\residual}[2]{#1 \uparrow #2}

\newcommand{\lift}[2]{#1|^{#2}}

\newcommand{\FinPaths}{\mathit{FinPaths}}
\newcommand{\Path}{\mathit{Path}}

\newcommand{\infpath}{\zeta}
\newcommand{\finpath}{\pi}
\newcommand{\fpath}{\finpath}
\newcommand{\ipath}{\infpath}

\newcommand{\last}{\mathit{last}}

\newcommand{\transformed}[1]{\hat{#1}}

\newcommand{\Repair}{\mathit{Rep}}
\newcommand{\Error}{\mathit{Err}}
\newcommand{\Operational}{\mathit{Op}}

\newcommand{\Avail}{\mathrm{Avail}}

\newcommand{\Nat}{\mathbb{N}}

\newcommand{\Rational}{\mathbb{Q}}

\newcommand{\neXt}{\bigcirc}
\DeclareMathOperator{\WeakUntil}{\ensuremath{\mathsf{W}}}
\DeclareMathOperator{\Until}{\ensuremath{\mathsf{U}}}
\DeclareMathOperator{\Eventually}{\ensuremath{\diamondsuit}}

\renewcommand{\>}{\rangle}

\newcommand{\goal}{\mathit{goal}}
 
\begin{document}

\maketitle

\begin{abstract}
Repair mechanisms are important within resilient systems to maintain the
system in an operational state after an error occurred. Usually, constraints on the
repair mechanisms are imposed, e.g., concerning the time or resources required (such as
energy consumption or other kinds of costs). 
For systems modeled by Markov decision processes (MDPs), 
we introduce the concept of \emph{resilient schedulers}, which represent control strategies 
guaranteeing that these constraints are always met within some given probability. 
Assigning rewards to the operational states of the system, we then aim towards resilient 
schedulers which maximize the long-run average reward, i.e., the expected mean payoff. 
We present a pseudo-polynomial algorithm that decides whether a resilient scheduler exists
and if so, yields an optimal resilient scheduler.
We show also that already the decision problem asking whether there exists a 
resilient scheduler is PSPACE-hard.

\end{abstract}

\section{Introduction}
\label{sec:intro}

Computer systems are resilient when they incorporate mechanisms to
adapt to changing conditions and to recover rapidly or at low costs from 
disruptions. The latter property of resilient systems is usually maintained
through repair mechanisms, which push the system towards an operational state
after some error occurred. 
Resilient systems and repair mechanisms have been widely studied in the literature
and are an active field of research (see, e.g., \cite{attoh2016resilience} for
an overview). 
Errors such as measurement errors, read/write errors, connection errors do not
necessarily impose a system error but may be repaired to foster the system
to be operational. 
Examples of repair mechanisms include rejuvenation
procedures that face the degradation of software over time \cite{german-book}, the evaluation
of checksums to repair communication errors, or methods to counter an attack
from outside a security system. 
The repair of a degraded software system could be achieved, e.g.,
by clearing caches (fast, very good availability), 
by running maintenance methods (more time, less availability, but higher success), 
or by a full restart (slow, cutting off availability, but guaranteed success).
Depending on the situation the system faces, there
is a trade-off between these characteristics and a choice has to be made, 
which of the repair mechanisms should be executed to fulfill further
constraints on the repair, which errors should be avoided, and to optimize an overall goal.
Usually, finding suitable control strategies performing the choices for
repair is done in an ad-hoc manner and requires a considerable
engineering effort.%

In this paper, we face the question of an automated synthesis of
\emph{resilient control strategies} that maximize the long-run average 
availability of the system. Inspired by the use of probabilistic response patterns to 
describe resilience \cite{Camara2012}, we focus on control strategies that
are \emph{probabilistically resilient}, i.e., with high probability repair 
mechanisms succeed within a given amount of time or other kinds of costs.
Our formal model we use to describe resilient systems is provided by 
Markov decision processes (MDPs, see, e.g., \cite{Puterman:book,Kallenberg}).
That is, directed graphs over states with edges annotated by actions 
that stand for non-deterministic choices and stochastic information about the probabilistic
choices resolved after taking some action. Following \cite{Baier2014,Huang2016}, 
we distinguish between three kinds of states: error, repair and 
operational states. Error states stand for states where
a disruption of the system is discovered, initiating a repair mechanism modeled by repair states.
Operational states are those states where the system is available and 
no repair is required. To reason about the trade-off between choosing control strategies,
we amend error and repair states with cost values, and operational states with 
payoff values, respectively. Assigned costs formalize, e.g., the time required 
or the energy consumed for leaving an error or repair state. Likewise, assigned payoff values
quantify the benefit of some operational state, e.g., stand for the number of 
successfully completed tasks while being operational.
We define the long-run average availability as the mean-payoff. 
Control strategies in MDPs are provided by (randomized) schedulers that, depending on the
history of the system execution, choose the probability of the next action to fire. When the
probabilities for action choices are Dirac, i.e., exactly one action is chosen almost surely,
the scheduler is called deterministic. Schedulers which select an action only depending on
the current state, i.e., do not depend on the history, are called memoryless. 
For a given cost bound $R$ and a probability threshold $\wp$, we call a scheduler
\emph{resilient} if the scheduler ensures for every error a recovery within at most 
$R$ costs with probability at least $\wp$.
\paragraph{Our Contribution.} We show that if the cost bound $R$ is represented in \emph{unary}, the existence of a resilient scheduler is solvable in polynomial time. Further, we show that if there is at least one resilient scheduler, then there also exists an \emph{optimal} resilient scheduler $\rsched$ computable in polynomial time. Here, optimality means that $\rsched$ achieves the maximal long-run average availability among all resilient schedulers. The constructed scheduler $\rsched$ is randomized and uses finite memory.  The example below illustrates that deterministic or memoryless randomized schedulers are less powerful. If $R$ is encoded in binary, our algorithms are exponential, and we show that deciding the existence of a resilient scheduler becomes PSPACE-hard. Let us note that all numerical constants (such as $\wp$ or MDP transition probabilities) except for $R$ are represented as fractions of binary numbers.
The key technical ingredients of our results are non-trivial observations about the structure of resilient schedulers, which connect the studied problems to the existing works on MDPs with multiple objectives and optimal strategy synthesis \cite{Kallenberg,EtessamiKVY08,BBCFK-TwoViews14}. The PSPACE-hardness result is obtained by a simple reduction of the cost-bounded reachability problem in acyclic MDPs \cite{HaaseKiefer15}. 
More details are given at appropriate places in Section~\ref{sec:algorithm} and in the appendix.

\begin{figure}[t]
 	\centering%
\definecolor{darkgreen}{rgb}{0.0, 0.5, 0.0}
\definecolor{darkgreen}{rgb}{0.0, 0.5, 0.0}
\begin{tikzpicture}[x=2cm,y=1cm,font=\scriptsize]
    \node[label={\textcolor{darkgreen}{0}}] (s0) at (0,0)   [op] {$\sinit$};
    \node[label={\textcolor{red}{0}}] (f) at (1,0)   [fail] {$error$};
    \node[label={\textcolor{darkgreen}{0}}] (o1) at (3,1)   [op] {$op_1$};
    \node[label={\textcolor{red}{1}}] (r) at (2,0)   [rep] {$rep$};
    \node[label={\textcolor{darkgreen}{1}}] (o2) at (3,-1)   [op] {$op_2$};

    \draw [tran] (s0) -- (f);
    \draw [tran] (f) -- (r);
    \draw [tran,rounded corners] (r) -- +(.2,-.6) -- node[below] {$\beta$, 1/2} +(-.2,-.6) -- (r);
    \draw [tran] (r) -- node[above] {\hspace{1.5em}$\beta$, 1/2} (o2);
    \draw [tran,rounded corners] (o2) -- +(.4,.3) --  +(.4,-.3) -- (o2); 
    \draw [tran,rounded corners] (o1) -- +(.4,.3) --  +(.4,-.3) -- (o1); 
    \draw [tran,rounded corners] (r) -- node[below] {$\alpha$} (o1);
\end{tikzpicture}  	\caption{\label{fig:randomized} Optimal resilient schedulers
 		might require finite memory and randomization}
 	\label{fig-sched}
\end{figure}
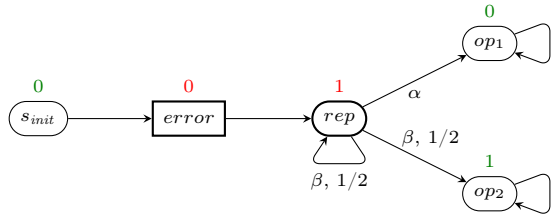
\paragraph{Example.} 
As a simple example, consider an MDP model of a resilient system depicted in Fig.~\ref{fig:randomized}.
Operational states are depicted by thin rounded boxes, 
error states are shown as rectangles and repair states 
are depicted by thick-rounded boxes. 
Assigned cost and payoff values are indicated above the nodes of the MDP.
For edges without any action name or probability, we assume one action with probability one.
The system starts its execution in the operational state $\sinit$, from which it reaches the error state $error$ and directly invokes a repair mechanism by switching to the  repair state $rep$, where either action $\alpha$ or $\beta$ can be chosen.
After taking $\alpha$, an operational state $op_1$ is reached that, however,
does not grant any payoff. 
When choosing $\beta$, a fair coin is flipped and either the repair mechanism has to be tried
again or the operational state $op_2$ is reached, while providing the payoff value 1 for each visit of $op_2$. 
Assume that we have given the cost bound $R=2$ and  probability threshold $\wp=4/5$.
The memoryless deterministic strategy always choosing $\beta$ yields the maximal possible mean payoff of $1$,
but is not resilient as $\wp>1-1/{2^R}=3/4$. 
The memoryless randomized scheduler that chooses $\beta$ with probability 
$2/\sqrt{5}$ is resilient and achieves the maximal mean payoff of 
$1/(\sqrt{5} - 1) \approx 0.809$, 
when ranging over all memoryless randomized schedulers. 
Differently, 
the finite-memory randomized scheduler playing $\beta$ with probability $4/5$ in the 
second step and with probability 1 in all other steps yields the 
mean payoff of $0.9$, which is optimal within all resilient schedulers. 
As this example shows, optimal resilient schedulers might require randomization
and finite memory in terms of remembering the accumulated costs 
spent so far after an error occurred.

\paragraph{Related work.} 
Concerning the analysis of resilient systems, \cite{Baier2014} presented algorithms to
reason about trade-offs between costs and payoffs using (probabilistic) model-checking techniques. 
In \cite{Longo2017}, several metrics to quantify resiliency and their applications to large scale
systems has been detailed.

Synthesis of control strategies for resilient systems have been mainly considered in the
non-probabilistic setting. In \cite{Huang2016}, a game-theoretic approach towards 
synthesizing strategies that maintain a certain resilience level has been presented.
The resilience level is defined in terms of the number of errors from which the 
system can recover simultaneously.
Automatic synthesis of Pareto-optimal implementations of resilient systems
were detailed in \cite{Ehlers2014}. Robust synthesis procedures with both,
qualitative and mean-payoff objectives have been presented in \cite{Bloem2014}.
In \cite{Girault2009}, the authors present algorithms to synthesize controllers
for fault-tolerant systems compliant to constraints on power consumption.

Optimization problems for MDPs with mean-payoff objectives and constraints on 
cost structures have been widely studied in the field of constrained Markov 
decision processes (see, e.g., \cite{Puterman:book} and
\cite{altman-constrainedMDP} for an overview). 
MDPs with multiple constraints on the probabilities for 
satisfying $\omega$-regular specifications 
were studied in~\cite{EtessamiKVY08}.
This work has been extended to also allow for (multiple) 
constraints on the expected total
reward in MDPs with rewards in~\cite{ForejtKNPQ11}.
Synthesis of optimal schedulers with multiple long-run
average objectives in MDPs has been considered in~\cite{Chatterjee2007,BBCFK-TwoViews14}.
All of the mentioned approaches have in common that they adapt well-known linear programs 
to synthesize optimal memoryless randomized schedulers (see, e.g., \cite{Kallenberg,Puterman:book}).
We also use combinations of similar techniques to find optimal resilient schedulers.
As far as we know, we are the first to consider mean-payoff optimization problems
under cost-bounded reachability probability constraints. Although we 
investigate these problems in the context of resilient systems, they are interesting
by its own.

\section{Notations and problem statement}
\label{sec:prelim}
Given a finite set $X$, we denote by $\dist(X)$ the set of \emph{probability distributions}
on $X$, i.e., the set of functions $\mu\colon X\rightarrow [0,1]$ where
$\sum_{x\in X} \mu(x) = 1$. 
By $X^\infty$ we denote finite or infinite sequences of elements of $X$.
We assume that the reader is familiar with principles about probabilistic
systems, logics, and model-checking techniques and refer to \cite{BK08}
for an introduction in these subjects.

\subsection{Markov decision processes}
\label{sec-mdp}
A \emph{Markov decision process (MDP)} is a triple $\cM=(S,\Act,P,\sinit)$,
where $S$ is a finite state space, $\sinit \in S$ an initial state,
$\Act$ a finite set of actions, and $P\colon S \times \Act \times S \to [0,1]$
a transition probability function, i.e., a function where 
$\sum_{s'\in S} P(s,\alpha,s')\in \{0,1\}$ 
for all $s\in S$ and $\alpha \in \Act$.
For $s\in S$, let $\Act(s)$ denote the set of 
actions $\alpha \in \Act$ that are enabled in~$s$, i.e.,
$\alpha \in \Act(s)$ iff
$P(s,\alpha,\cdot)$ is a probability distribution over $S$.
Unless stated differently, we suppose that any MDP does not have any trap states,
i.e., states $s$ where $\Act(s) = \varnothing$.
\emph{Paths} in $\cM$ are alternating sequences 
$s_0\alpha_0s_1\alpha_1\ldots\in S\times(\Act{\times}S)^\infty$ 
of states and actions,
such that $P(s_{i},\alpha_{i},s_{i+1})>0$ for all $i\in\Nat$.
The set of all finite paths starting in state $s\in S$ is denoted by 
$\FinPaths(s)$, where we omit $s$ when all finite paths from any state
are issued. 
A \emph{(randomized, history-dependent) scheduler} for $\cM$ is a function 
$\sched\colon \FinPaths\rightarrow\dist(\Act)$. 
A \emph{$\sched$-path} in $\cM$ is a
path $\pi=s_0\act_0s_1\act_1\ldots$ in $\cM$ where for all $n\in\Nat$
we have that $\sched(s_0\act_0s_1\act_1\ldots\act_{n-1}s_n)(\act_n)>0$.
We write
$\Pr\nolimits^{\sched}_{\cM,s}$ for the probability measure on infinite paths of $\cM$
induced by a scheduler $\sched$ and starting in $s$.
For a scheduler $\sched$ and $\fpath\in\FinPaths$, 
$\residual{\sched}{\fpath}$ denotes the \emph{residual scheduler}
$\tsched$ given by $\tsched(\fpath') = \sched(\fpath ; \fpath')$
for each finite path $\fpath'$ where 
the first state of $\fpath'$ equals the last state of $\fpath$. 
Here $;$ is used for the concatenation operator on finite paths.
$\sched$ is called \emph{memoryless} if
$\sched(s)=\sched(\pi)$
for all $s\in S$ and all finite paths $\pi\in\FinPaths$ 
where the last state of $\pi$ is $s$.
We abbreviate memoryless (randomized) schedulers
as \emph{MR-schedulers}. 

\subsection{Markov decision processes with repair}
\label{sec-MDP-repair}
Let $\cM=(S,\Act,P,\sinit)$ be an MDP and suppose that 
we have given two disjoint sets of states $\Error,\Operational\subseteq S$. 
Intuitively, $\Error$ stands for the set of states where an error occurs,
and $\Operational$ stands for the set of states where the system modeled is operational.
In all other states, we assume that a repair mechanism is running, triggered directly
within the next transition after some error occurred.
We formalize the latter assumption by
\begin{equation}
  \label{weakuntil}
  e \models \forall \neXt \forall ( \neg \Error \WeakUntil \Operational)
  \qquad
  \text{for all states $e\in \Error$}
  \tag{*}
\end{equation}
where $\neXt$ and $\WeakUntil$ stand for the standard next and weak-until 
operator, respectively, borrowed from computation tree logic (CTL, see, e.g.,
\cite{BK08}). 
Assumption~\eqref{weakuntil} also asserts
that as soon as a repair protocol has been started,
the system does not enter a new error state before a successful repair, i.e.,
until the system switches to its operational mode. 
Further, we suppose that states in $\cM$ are amended with non-negative integer values,
i.e., we are given a non-negative integer reward function \mbox{$\rew\colon S \to \Nat$}.
For an operational state $s\in\Operational$, the value
$\rew(s)$ is viewed as the \emph{payoff} value of state $s$,
while for the non-operational states $s\in S{\setminus}\Operational$,
the value $\rew(s)$ is viewed as the repairing \emph{costs} caused by
state $s$.
To reflect this intuitive meaning of the reward values, we shall
write $\payoff(s)$ instead of $\rew(s)$ for $s\in \Operational$
and $\cost(s)$ instead of $\rew(s)$ for $s\in S{ \setminus} \Operational$.
Furthermore, we assume $\payoff(s)=0$ if $s\in S{\setminus}\Operational$ and 
$\cost(s)=0$ if $s \in \Operational$.
For a finite path $\pi= s_0\alpha_0s_1 \ldots \alpha_{n-1}s_n$, let $\cost(\pi)$ and $\payoff(\pi)$ 
be $\sum_{i=0}^{n} \cost(s_i)$ and $\sum_{i=0}^{n} \payoff(s_i)$, respectively.

An \emph{MDP with repair} is formally defined as a tuple $(\cM,\Error,\Operational,\rew)$,
where  assumption \eqref{weakuntil} is satisfied and the transition probability function 
of $\cM$ is rational, assuming representation of probabilities as fractions of binary numbers. 
\subsection{Long-run availability and resilient schedulers}
Given an MDP with repair $(\cM,\Error,\Operational,\rew)$ and a scheduler $\sched$ for $\cM$, we define the \emph{long-run availability} of $\sched$, denoted by $\Avail^{\sched}_{\cM,\sinit}$, 
as the expected long-run average (mean-payoff) of the payoff function.
That is,  for any $s_0\in S$, $\Avail^{\sched}_{\cM,s_0}$ agrees with the expectation of the
random variable $X$ under $\Pr\nolimits^{\sched}_{\cM,s_0}$ 
that assigns to each infinite path 
$\infpath = s_0 \, \act_0 \, s_1 \, \act_1 \, s_2 \, \act_2 \ldots$ 
the value
\[
   X(\infpath) \ \ = \ \ 
   \liminf_{n \to \infty} \
     \frac{1}{n} \sum_{i=0}^{n-1} \payoff(s_i).
\]
Let us further assume that we have given a rational probability threshold $\wp\in (0,1]$
and a cost bound $R\in \Nat$. The threshold $\wp$ is always represented as a fraction of two binary numbers. The bound $R$ is represented either in binary or in unary, which significantly influences the (computational) complexity of the studied problems.

\begin{definition}[Resilient schedulers]
\label{def:res-scheduler} 
A scheduler $\sched$ is said to be \emph{probabilistically resilient} 
with respect to $\wp$ and $R$ if the following conditions
\eqref{res-sched} and \eqref{res-rep} hold for all finite $\sched$-paths $\fpath$ from $\sinit$ to an error state~$s$:
\begin{align}
  \label{res-sched}
  \Pr\nolimits^{\residual{\sched}{\fpath}}_{\cM,s}
        \bigl(\ \Eventually^{\leqslant R} \Operational \ \bigr) 
  \ \ &\geqslant \ \ \wp
  \tag{\textrm{Res}} \\
\label{res-rep}
\Pr\nolimits^{\residual{\sched}{\fpath}}_{\cM,s}
\bigl(\ \Eventually \Operational \ \bigr) 
\ \ &= \ \ 1
\tag{\textrm{ASRep}}
\end{align}
Here, $\Eventually \Operational$ denotes the set of
infinite paths $\ipath$ for which there exist a finite path~$\fpath'$ and an infinite path 
$\varrho$ such that $\ipath=\fpath';\varrho$ and the last state of $\fpath'$ is
in $\Operational$. %
Further, $\Eventually^{\leqslant R} \Operational$  denotes the set $\Eventually \Operational$ restricted to paths satisfying $\cost(\fpath')\leqslant R$.
\end{definition}

The task addressed in this paper is to check the existence of 
resilient schedulers (i.e., schedulers that are probabilistically resilient w.r.t. $\wp$ and $R$), 
and if so, construct an \emph{optimal} resilient scheduler $\rsched$ that has maximal
long-run availability amongst all resilient schedulers, i.e., 
$\Avail^{\rsched}_{\cM,\sinit}=\Avail^{\max}_{\cM,\sinit}$,
where
\[
   \Avail^{\max}_{\cM,\sinit} \ \ = \ \ 
   \sup  \ 
   \bigl\{ \ \Avail^{\rsched'}_{\cM,\sinit} \ : \ 
             \text{$\rsched'$ is a resilient scheduler} \ 
   \bigr\}.
\]
\section{The results}
\label{sec:algorithm}

In the following, we present and prove our main result of this paper:

\begin{theorem}
\label{thm-main}
	Let $(\cM,\Error,\Operational,\rew)$ be an MDP with repair, $\wp\in \ (0,1]$ a rational probability threshold, and $R \in \Nat$ a cost bound encoded in unary. The existence of a probabilistically resilient scheduler w.r.t.{} $\wp$ and $R$ is decidable in polynomial time. If such a scheduler exists, then an \emph{optimal} probabilistically resilient scheduler $\rsched$ (w.r.t.{} $\wp$ and $R$) is computable in polynomial time.
\end{theorem}
If $R$~is encoded in binary, our algorithms are exponential, and we show that even the existence of a probabilistically resilient scheduler w.r.t.{} $\wp$ and $R$ becomes PSPACE-hard. The optimal scheduler $\rsched$ is randomized and history dependent, which is unavoidable (see the example in the introduction%
). More precisely, the memory requirements of $\rsched$ are finite with at most $|\Error| \cdot R$ memory elements, and this memory is only used in the repairing phase where the scheduler needs to remember the error state and the total costs accumulated since visiting this error state.

For the rest of this section, we fix an MDP with repair $(\cM,\Error,\Operational,\rew)$ where $\cM=(S,\Act,P,\sinit)$, a rational probability threshold $\wp\in \ (0,1]$, and a cost bound $R \in \Nat$. We say that a scheduler is \emph{resilient} if it is probabilistically resilient w.r.t.{} $\wp$ and $R$. 

The proof of Theorem~\ref{thm-main} is obtained in two steps. First, the MDP $\cM$ is transformed into a suitable MDP $\transformed{\cM}$ where the total costs accumulated since the last error are explicitly remembered in the states. 
Hence, the size of $\transformed{\cM}$ is polynomial in the input size if $R$ is encoded in unary. 
We will show that the problem of computing an optimal resilient scheduler can be safely considered in $\transformed{\cM}$ instead of $\cM$. 
In the second step, it is shown that there exists an optimal \emph{memoryless} resilient scheduler for $\transformed{\cM}$ computable in time polynomial in the size of $\transformed{\cM}$. This is the very core of our paper requiring non-trivial observations and constructions. Roughly speaking, we start by connecting our problem to the problem of multiple mean-payoff optimization, and use the results and algorithms presented in \cite{BBCFK-TwoViews14} to analyze the limit behavior of resilient schedulers. First, we show how to compute the set of end components such that resilient schedulers can stay only in these end components without loosing availability. We also compute memoryless schedulers for these end components that can safely be adopted by resilient schedulers. Then, we show that the behavior of a resilient scheduler prior entering an end component can also be modified so that it becomes memoryless and the achieved availability does not decrease. After understanding the structure of resilient schedulers, we can compute an optimal memoryless resilient scheduler for $\transformed{\cM}$ by solving suitable linear programs.

The first step (i.e, the transformation of $\cM$ into $\transformed{\cM}$) is described in Section~\ref{sec:transformation}, and the second step in Section~\ref{sec-res-solving}.

\subsection{Transformation}

\label{sec:transformation}

Let $(\transformed{\cM},\transformed{\Error},\transformed{\Operational},\transformed{\rew})$ be an MDP with repair where $\transformed{\cM}$ is an MDP 
$(\transformed{S},\transformed{\Act}, \transformed{P},\sinit)$ such that $\transformed{S} =  S  \cup \Repair$ with
\[
  \Repair \ \ = \ \ 
  \Error \times S \times \{0,1,\ldots,R\}.
\]
Intuitively, state $\<e,s,r\>\in\Repair$ indicates that the system is in state $s$
executing a repair procedure that has been triggered by visiting $e\in\Error$ 
somewhen in the past and with accumulated costs $r$ so far.
For technical reasons, we also include triples $\<e,s,r\>$ with $s\in \Operational$
in which case a repair mode with total cost $r$ has just finished.
The sets of error and operational states in $\transformed{\cM}$ are:
\begin{center}
   $\transformed{\Error} \ = \ \Error$ \ \ \ and \ \ \
   $\transformed{\Operational} \ = \ 
    \Operational \cup 
    \bigl\{\, \<e,s,r\> \in \Repair \, : \, s\in \Operational \, \bigr\}$.
\end{center}
The action set of $\transformed{\cM}$ is the same as for $\cM$.
In what follows, 
we write $\transformed{\Act}(\transformed{s})$ 
for the set of actions that are enabled
in state $\transformed{s}$ of $\transformed{\cM}$.
Then,
$\transformed{\Act}(s)=\transformed{\Act}(\<e,s,r\>)=\Act(s)$.
Let $s,s'\in S$ and $\alpha \in \Act$.
Then,
$\transformed{P}(s,\alpha,s') = P(s,\alpha,s')$ if $s \notin \Error$.
If $e\in \Error$ and $\alpha\in \Act(e)$, then
$$
   \transformed{P}\big(e,\alpha,\<e,s,\cost(e)\>\big) \ =\ P(e,\alpha,s)
$$
For, $e \in \Error$, $r\in \{0,1,\ldots,R\}$, and $\alpha \in \Act(s)$ we have:
$$
 \begin{array}{rcll}
  \transformed{P}\big(\<e,s,r\>,\alpha,\<e,s',r{+}\cost(s)\>\big) & \ = \ &
  P(s,\alpha,s')
  &
  \quad \text{if $r{+}\cost(s) \leqslant R$ and $s\notin \Operational$}
  \\[1.2ex]

  \transformed{P}\big(\<e,s,r\>,\alpha,s'\big) & \ = \ & P(s,\alpha,s')
  &
  \quad \text{if $r{+}\cost(s) > R$ or $s\in \Operational$}
 \end{array}
$$
In all remaining cases, we set $\transformed{P}(\cdot)=0$.
The reward function $\transformed{\rew}$ of $\transformed{\cM}$ is given by
$\transformed{\cost}(s)=\transformed{\cost}(\<e,s,r\>) = \cost(s)$
and $\transformed{\payoff}(s)= \transformed{\payoff}(\<e,s,r\>) = \payoff(s)$.
Note that assumption \eqref{weakuntil} ensures that $s \notin \Error$
for all states $\<e,s,r\>$.

There is a one-to-one correspondence between the paths in $\cM$ and in
$\transformed{\cM}$. More precisely,
given a (finite or infinite) path $\transformed{\fpath}$ 
in $\transformed{\cM}$, let
$\transformed{\fpath}|_{\cM}$ denote the unique path in $\cM$ that arises from
$\transformed{\fpath}$ by replacing 
each repair state $\<e,s,r\>$ with $s$.
Vice versa, each path $\fpath$ in $\cM$ can be lifted to a path
$\lift{\fpath}{\transformed{\cM}}$ in $\transformed{\cM}$ 
such that $(\lift{\fpath}{\transformed{\cM}})|_{\cM}=\fpath$.
Next lemmas follow directly from definitions of $\transformed{\cost}$ and $\transformed{\payoff}$.

\begin{lemma}
 \label{cost-fpath}
  For each finite path $\transformed{\fpath}$ in $\transformed{\cM}$ starting
  in some state $e\in \transformed{\Error}$ we have
  $\transformed{\cost}(\transformed{\fpath})=
      \cost(\transformed{\fpath}|_{\cM})$.
\end{lemma}

\begin{lemma}
	\label{payoff-infpath}
	For each infinite path $\transformed{\infpath}$ in $\transformed{\cM}$, 
	$\transformed{\payoff}(\transformed{\infpath})=
	\payoff(\transformed{\infpath}|_{\cM})$.
\end{lemma}

The one-to-one correspondence between the paths in $\cM$ 
and in $\transformed{\cM}$
carries over to the schedulers for $\cM$ and $\transformed{\cM}$.
Given a scheduler $\sched$ for $\cM$, let $\lift{\sched}{\transformed{\cM}}$
denote the scheduler for $\transformed{\cM}$ given by 
$\lift{\sched}{\transformed{\cM}}(\transformed{\fpath}) = 
 \sched(\transformed{\fpath}|_{\cM})$
for all finite paths $\transformed{\fpath}$ of $\transformed{\cM}$.
This yields a scheduler transformation 
$\sched \mapsto \lift{\sched}{\transformed{\cM}}$ 
that maps each scheduler for $\cM$ to a scheduler for $\transformed{\cM}$.
Vice versa, given a scheduler $\transformed{\sched}$ 
for $\transformed{\cM}$ there
exists a scheduler $\transformed{\sched}|_{\cM}$ such that
$\transformed{\sched}= 
 \lift{(\transformed{\sched}|_{\cM})}{\transformed{\cM}}$.

Due to assumption \eqref{weakuntil} we have that
  $s\notin \Error$ for all repair states $\<e,s,r\>$
  that are reachable from $e$ in $\transformed{\cM}$. 
Thus, with Lemma \ref{cost-fpath} and Lemma \ref{payoff-infpath},
we obtain:

\begin{lemma}
 \label{sched-schedprime}
   Let $\sched$ be a scheduler for $\cM$ and 
   $\transformed{\sched}$ a scheduler for $\transformed{\cM}$ such that
   $\sched = \transformed{\sched}|_{\cM}$.
   Then: 
   \begin{enumerate}
   \item [(a)] 
      For each state $e\in \Error$: \
      $\Pr\nolimits^{\sched}_{\cM,e}
       \bigl(\, \Eventually \Operational \, \bigr)
       \ = \ 
       \Pr\nolimits^{\transformed{\sched}}_{\transformed{\cM},e}
       \bigl(\, \Eventually \transformed{\Operational} \, \bigr) 
      $
      and
      $$
      \Pr\nolimits^{\sched}_{\cM,e}
      \bigl(\, \Eventually^{\leqslant R}\Operational \, \bigr)
      \ \ = \ \ 
      \Pr\nolimits^{\transformed{\sched}}_{\transformed{\cM},e}
      \bigl(\, \Eventually^{\leqslant R}\transformed{\Operational} \, \bigr)
      \ \ = \ \ 
      \Pr\nolimits^{\transformed{\sched}}_{\transformed{\cM},e}
      \bigl(\, \neXt (\Repair \Until \Operational_e) \, \bigr)
      $$
      where 
      $\Operational_e \, = \, 
      \bigl\{ \, \<e,s,r\>\in \Repair \, : \, s\in \Operational \, \bigr\}$.
  \item [(b)] 
     $\Avail^{\sched}_{\cM,\sinit}
      \ =  \ 
      \Avail^{\transformed{\sched}}_{\transformed{\cM},\sinit}$
   \end{enumerate}
\end{lemma}

\begin{corollary}
\label{cor-transform}
     $\Avail^{\max}_{\cM,\sinit}
      \ = \ 
      \Avail^{\max}_{\transformed{\cM},\sinit}$
\end{corollary}

\begin{proof}
The above transformations $\fpath \mapsto \lift{\fpath}{\transformed{\cM}}$ 
and  $\sched \mapsto \lift{\sched}{\transformed{\cM}}$ for 
paths and schedulers of $\cM$
to paths and schedulers of $\transformed{\cM}$, and the inverse mappings
$\transformed{\fpath} \mapsto \transformed{\fpath}|_{\cM}$ and
$\transformed{\sched} \mapsto \transformed{\sched}|_{\cM}$
for paths and schedulers of $\transformed{\cM}$
to paths and schedulers of $\cM$ are compatible with the
residual operator for schedulers in the following sense:
$$
  \lift{(\residual{\sched}{\fpath})}{\transformed{\cM}}  
  \ =  \ 
  \residual{\bigl(\lift{\sched}{\transformed{\cM}}\bigr)}
           {\bigl( \lift{\fpath}{\transformed{\cM}}\big)}
  \qquad \text{and} \qquad
  \big(\residual{\transformed{\sched}}{\transformed{\fpath}}\big)|_{\cM} 
   \ =  \ 
  \residual{\big(\transformed{\sched}|_{\cM}\big)}
           {\big(\transformed{\fpath}|_{\cM}\big)}
$$
Thus, part (a) of Lemma \ref{sched-schedprime} yields that 
$\transformed{\sched}$
is resilient for $\transformed{\cM}$ if and only if $\sched$ 
is resilient for $\cM$.
Part (b) of Lemma \ref{sched-schedprime} then yields the claim.\qed
\end{proof}

The following mainly technical lemma shows that residual schedulers 
arising from resilient schedulers maintain the resilience property.

\begin{lemma}
	\label{lem-res-switch}
	Let $\sched$ be a resilient scheduler for $\transformed{\cM}$, and let $s$ be a state of $\transformed{\cM}$ such that $s \not\in \Repair$. Let $\cP$ be a set of finite $\sched$-paths initiated in $\sinit$ and terminating in~$s$, and let $\sched'$ be a scheduler for $\transformed{\cM}$  resilient for the initial state  changed to~$s$. Consider the scheduler $\sched[\cP,\sched']$ which is the same as $\sched$ except that for every finite path $w$ such that $w = w';w''$ where $w' \in \cP$ we have that  $\sched[\cP,\sched'](w) = \sched'(w'')$. Then $\sched[\cP,\sched']$ is resilient (for the initial state $\sinit$).
\end{lemma}

\subsection{Solving the resilience-availability problem for $\transformed{\cM}$}
\label{sec-res-solving}
In this section, we analyze the structure of resilient schedulers for $\transformed{\cM}$ and prove the following proposition:

\begin{proposition}
\label{prop-trans-solve}
	The existence of a resilient scheduler for $\transformed{\cM}$ can be decided in polynomial time. The existence of some resilient scheduler for $\transformed{\cM}$ implies the existence of an optimal memoryless resilient scheduler for $\transformed{\cM}$ computable in polynomial time. 
\end{proposition}
Note that Theorem~\ref{thm-main} follows immediately from 
Proposition~\ref{prop-trans-solve} and Corollary~\ref{cor-transform}.

We start by introducing some notions. A \emph{fragment} of $\transformed{\cM} = (\transformed{S},\transformed{\Act}, \transformed{P},\sinit)$ is a pair $(F,\A)$ where $F \subseteq \transformed{S}$ and $\A\colon F \rightarrow 2^{\transformed{\Act}}$ is a function such that $\A(s) \neq \varnothing$ and $\A(s) \subseteq \transformed{\Act}(s)$ for every $s \in F$. 
An \emph{MR-scheduler for $(F,\A)$} is a function $\sched_F$ assigning a probability distribution over $\A(s)$ to every $s \in F$. We say that a scheduler $\sched$ for $\transformed{\cM}$ is \emph{consistent} with $\sched_F$ if for every $\fpath\in\FinPaths$ ending in a state of $F$ we have that $\sched(\pi) = \sched_{F}(\pi)$.

An \emph{end component} of $\transformed{\cM}$ is a fragment $(E,\A)$ of $\transformed{\cM}$ such that 
\begin{itemize}
\item 
  $(E,\A)$ is strongly connected, i.e., for all $s,s' \in E$ there is a finite path $s_0\alpha_0 s_1\ldots \alpha_{n-1} s_n$ from $s=s_0$ to $s'=s_n$ such that $s_i \in E$ and $\alpha_i \in \A(s_{i})$ for all $0\leq i < n$;
\item 
 for all $s \in E$, $\alpha \in \A(s)$, and $s' \in \transformed{S}$ such that $\transformed{P}(s,\alpha,s') > 0$ we have $s' \in E$.
\end{itemize}
Let $\sched$ be a scheduler for $\transformed{\cM}$ (not necessarily resilient). For every infinite path~$\infpath$, let $F_\infpath$ be the set of states occurring infinitely often in~$\infpath$. For every $s \in F_\infpath$, let $\A_\infpath(s)$ be the set of all actions executed infinitely often from $s$ along~$\infpath$. For a fragment $(F,\A)$, let $\Path(F,\A)$ be the set of all infinite paths $\infpath$ such that $F_{\infpath} = F$ and $\A_{\infpath} = \A$, 
and let $\Pr\nolimits^{\sched}_{\transformed{\cM},\sinit}(F,\A)$ be the probability of all $\infpath \in \Path(F,\A)$ starting in $\sinit$. If $(F,\A)$ is \emph{not} an end component, then clearly $\Pr\nolimits^{\sched}_{\transformed{\cM},\sinit}(F,\A) = 0$. Hence, there are end components $(F_1,\A_1),\ldots,(F_m,\A_m)$ such that:
\begin{center}
   $\Pr\nolimits^{\sched}_{\transformed{\cM},\sinit}(F_i,\A_i) > 0$ 
   for all $i \leq m$, and 
   $\sum\limits_{i=1}^m  
     \Pr\nolimits^{\sched}_{\transformed{\cM},\sinit}(F_i,\A_i) = 1$ 
\end{center}
We say that $\sched$ \emph{stays} in these end components.

Proposition~\ref{prop-trans-solve} is proved as follows. We show that there is a set $\cE$, computable in time polynomial in $|\transformed{\cM}|$, consisting of triples of the form  $(E,\A,\sched_E)$ such that $(E,\A)$ is an end component of $\transformed{\cM}$ and $\sched_E$ is an MR-scheduler for $(E,\A)$, satisfying the following conditions (E1) and (E2):
\begin{description}
	\item [(E1)] If $(E,\A,\sched_E),(E',\A',\sched_{E'}) \in \cE$, 
  	then the two triples are either the same or $E \cap E' = \varnothing$. 
  	\item [(E2)]
          Every $(E,\A,\sched_E) \in \cE$ is \emph{strongly connected}, i.e., the directed graph $(E,\to)$, where $s \to s'$ iff
 there is some $\alpha \in \A(s)$ such that 
 $\sched_E(s)(\alpha) > 0$ and $\transformed{P}(s,\alpha,s')>0$, 
 is strongly connected.  
 (In this case, $E$ is a bottom strongly connected component of 
  the Markov chain induced by $\sched_E$.)
\end{description}
Further, we can safely restrict ourselves to resilient schedulers whose long-run behavior is captured by some subset $\cE' \subseteq \cE$ in the following sense:  
\begin{lemma}
\label{lem-claim-one}
Given the set $\cE$, %
for every resilient scheduler $\rsched$ there exist a set
$\cE' \subseteq \cE$ and a resilient scheduler $\rsched'$ such that
\begin{itemize}
	\item almost all $\rsched'$-paths starting in $\sinit$ visit a state of $\bigcup_{(E,\A,\sched_E) \in \cE'} E$,
	\item $\rsched'$ is consistent with $\sched_E$ for every $(E,\A,\sched_E) \in \cE'$,
	\item $\Avail^{\rsched}_{\transformed{\cM},\sinit} \ \leq \ \Avail^{\rsched'}_{\transformed{\cM},\sinit}$.
\end{itemize}	
\end{lemma}
Using Lemma~\ref{lem-claim-one}, we prove the following:
\begin{lemma}
\label{lem-claim-two}
Given the set $\cE$,  %
there is a linear program $\cL$ computable in time polynomial in $|\transformed{\cM}|$ satisfying the following: If $\cL$ is not feasible, then there is no resilient scheduler for $\transformed{\cM}$. Otherwise, there is a subset $\cE' \subseteq \cE$ and an MR-scheduler $\sched_{F}$ for the fragment $(F,\A)$ with $F = \transformed{S} \setminus \bigcup_{(E,\A,\sched_E) \in \cE'} E$ and $\A(s) =  \transformed{\Act}(s)$ for every $s \in F$ such that
\begin{itemize}
	\item $\cE'$ and $\sched_{F}$ are computable in time polynomial in $|\transformed{\cM}|$,
	\item the scheduler $\rsched$ consistent with $\sched_{F}$ and $\sched_{E}$ for every $(E,\A,\sched_E) \in \cE'$
		is resilient, and
	\item for every resilient scheduler $\rsched'$ we have that  
	$\Avail^{\rsched}_{\transformed{\cM},\sinit} \ \geq \ \Avail^{\rsched'}_{\transformed{\cM},\sinit}$.
\end{itemize}	
\end{lemma}

In the next subsections, we show how to compute the set $\cE$ satisfying conditions (E1) and (E2) in polynomial time and provide proofs for Lemmas~\ref{lem-claim-one} and~\ref{lem-claim-two}.
Note that Proposition~\ref{prop-trans-solve} then follows from Lemma~\ref{lem-claim-two} and the polynomial-time computability of $\cE$.

\subsubsection{Constructing the set $\cE$.}
For each $e\in \Error$, we define the weight function $\wgt_e\colon \transformed{S}\to \Rational$ given by
\[
\begin{array}{lcll}
\wgt_e(\<e,s,r\>) & = & 1{-}\wp \quad
&
\text{if $s \in \Operational$}
\\[1.5ex]

\wgt_e(\<e,s,r\>) & = & - \wp 
&
\text{if $s \notin \Operational$ and $r{+}\cost(s) > R$}
\end{array}
\]
and $\wgt_e(\transformed{s})=0$ otherwise (in particular, 
for all states in $\transformed{s} \in \transformed{S}$ that do not have the form $\<e,s,r\>$). For every scheduler $\sched$, let $\fMP_e^\sched$ be the expected value (under $\Pr\nolimits^{\sched}_{\transformed{\cM},\sinit}$) of the random variable $X_e$ assigning to each infinite $\sched$-path $\infpath = s_0 \, \act_0 \, s_1 \, \act_1 \, s_2 \, \act_2 \ldots$
the value
\begin{center}
   $X_e(\infpath) \ \ = \ \ 
   \liminf\limits_{n \to \infty} \
       \frac{1}{n} \sum\limits_{i=0}^{n-1} \wgt_e(s_i)$.
\end{center}
We say that a scheduler $\sched$ for $\transformed{\cM}$ is \emph{average-resilient} if $\fMP_e^\sched \geq 0$ for all $e \in \Error$. Note that if $\rsched$ is a resilient scheduler for $\transformed{\cM}$, then $X_e(\infpath) \geq 0$ for almost all $\infpath$ (this follows by a straightforward application of the strong law of large numbers). Thus, we obtain:

\begin{lemma}
	\label{lem-average-res}
	Every resilient scheduler for $\transformed{\cM}$ is average-resilient. 
\end{lemma}

Although an average-resilient scheduler for $\transformed{\cM}$ is not necessarily resilient, we show that the problems of maximizing the long-run availability under resilient and average-resilient schedulers are to some extent related. The latter problem can be solved by the algorithm of \cite{BBCFK-TwoViews14}. More precisely, by Theorem~4.1 of \cite{BBCFK-TwoViews14}, one can compute a linear program~$\cL_{\transformed{\cM}}$ in time polynomial in $|\transformed{\cM}|$ such that:
\begin{itemize}
\item 
  if $\cL_{\transformed{\cM}}$ is not feasible, then there is no average-resilient scheduler for $\transformed{\cM}$;
\item 
  otherwise, there is a 2-memory stochastic update scheduler $\asched$ for $\transformed{\cM}$,  constructible in time polynomial in $|\transformed{\cM}|$, which is average-resilient and achieves the maximal long-run availability among all average-resilient schedulers. 
\end{itemize}
The scheduler $\asched$ almost surely ``switches'' from its initial mode to its second mode where it behaves memoryless. Hence, there is a set $\cE_{\asched}$ (computable in time polynomial in $|\transformed{\cM}|$) comprising triples $(E,\A,\asched_E)$ that enjoy the following properties (H1) and (H2):
\begin{description}
	\item[(H1)] $(E,\A)$ is an end component of $\transformed{\cM}$ and $\asched_E$ is an MR-scheduler for   
	$(E,\A)$ achieving the \emph{maximal} long-run availability among all average-resilient schedulers for every initial state~$s \in E$. 
	\item[(H2)] If $(E,\A,\asched_E),(E',\A',\asched_{E'}) \in \cE_{\asched}$, then the two triples are either the same or $E \cap E' = \varnothing$. Further,
    every $(E,\A,\asched_E) \in \cE_{\asched}$ is strongly connected.
\end{description}
We show that for every $(E,\A,\asched_E) \in \cE_{\asched}$ and every $s \in E$, the scheduler $\asched_E$ is \emph{resilient} when the initial state is changed to~$s$ (see Lemma~\ref{lem-resilient}). So, $\asched$ starts to behave like a resilient scheduler after a ``switch'' to some 
$(E,\A,\asched_E) \in \cE_{\asched}$. However, in the initial transient phase, $\asched$ may violate the resilience condition, which may disallow a resilient scheduler $\rsched$ to  enter some of the end components of $\cE_{\asched}$. Thus, a resilient scheduler $\rsched$ can in general be forced to stay in an end component that does not appear in $\cE_{\asched}$. So, the set $\cE$ needs to be \emph{larger} than $\cE_{\asched}$, and we show that a sufficiently large $\cE$ is computable in polynomial time by Algorithm~\ref{alg-comp-compute}. 

Algorithm~\ref{alg-comp-compute} starts by initializing $\cQ$  to  $\transformed{\cM}$, $s$ to $\sinit$, and $\cE$ to~$\emptyset$. 
Then, it computes the linear program $\cL_{\cQ}$ and checks its feasibility. If $\cL_{\cQ}$ is not feasible, the initial state $s$ of $\cQ$ is removed from $\cQ$ in the way described below.
Otherwise, the algorithm constructs the scheduler $\asched$, adds $\cE_\asched$ to $\cE$, and ``prunes'' $\cQ$ into $\cQ \ominus \cE_\asched$. If the state $s$ is deleted from $\cQ$, some state of $\cQ$ is chosen as a new initial state. This goes on until $\cQ$ becomes empty.
Here, the MDP $\cQ \ominus X$ is the largest MDP subsumed by $\cQ$ which does not contain the states in $X \subseteq \transformed{S}$. 
Note that when a state of $\cQ$ is deleted, all actions leading to this state must be disabled; and if all outgoing actions of a state $s$ are disabled, then $s$ must be deleted. 
Hence, deleting the states appearing in $\cE_\asched$ may enforce deleting additional states and disabling further actions.
Note that every $(E,\A,\asched_E) \in \cE$ is obtained in some iteration of the repeat-until cycle of Algorithm~\ref{alg-comp-compute} by constructing the scheduler $\asched$ for the current value of $\cQ$.  
We denote this MDP $\cQ$ as ${\cQ}_{E}$ (note that ${\cQ}_{E}$ is not necessarily connected).
The set $\cE$ returned by  Algorithm \ref{alg-comp-compute}
indeed satisfies conditions (E1) and (E2).
The outcome $\cE=\varnothing$ is possible, in which case there is no resilient scheduler for $\transformed{\cM}$ as the linear program $\cL$ of Lemma~\ref{lem-claim-two} is not feasible for $\cE = \varnothing$.

An immediate consequence of property~(H1) is the following:

\begin{lemma}
	\label{lem-E-max}
	Let $(E,\A,\asched_E) \in \cE$ and $s \in E$. Then $\asched_E$ achieves the \emph{maximal} long-run availability for the initial state $s$ among all average-resilient schedulers for ${\cQ}_E$. 
\end{lemma}

\begin{algorithm}[t]
	\SetAlgoLined
	\DontPrintSemicolon
	\SetKwInOut{Input}{input}\SetKwInOut{Output}{output}
	\SetKwData{n}{n}\SetKwData{f}{f}\SetKwData{g}{g}
	\SetKwData{Low}{l}\SetKwData{x}{x}
	\Input{the transformed MDP $\transformed{\cM}$} 
	\Output{the set $\cE$ satisfying (E1) and (E2)}
	\BlankLine
	$\cQ :=  \transformed{\cM}$, $s := \sinit$, $\cE := \emptyset$\;
	\Repeat{$\cQ$ becomes empty} {
		Compute the linear program $\cL_{\cQ}$\;  
		\eIf{$\cL_{\cQ}$ is feasible}
		{   compute the scheduler $\asched$ and the set $\cE_\asched$ satisfying (H1) and (H2)\;
			$\cE := \cE \cup \cE_{\asched}$ \;
			$\cQ := \cQ \ominus \cE_{\asched}$ \;
		}
		{$\cQ := \cQ \ominus \{s\}$ \; }
		\If{$s \text{ is not a state of } \cQ$}{$s := \text{ some state of } \cQ$ \; }
	}     
	\Return $\cE$ \;	
	\caption{Computing the set $\cE$.}
	\label{alg-comp-compute}

\end{algorithm}
\noindent
The next lemma follows easily from the construction of $\cE$.
\begin{lemma}
	\label{lem-cover} 
	Let $\sched$ be a scheduler for $\transformed{\cM}$ (not necessarily resilient) and let $(F,\cB)$ be an end component where $\sched$ stays with positive probability. Then there is   $(E,\A,\asched_E) \in \cE$ such that $(F,\cB)$ is an end component of $\cQ_E$ and $F \cap E \neq \varnothing$.
\end{lemma}	

\noindent
Let $(E,\A,\asched_E) \in \cE$. Since $\asched_E$ is an MR-scheduler, the behavior of $\asched_E$ in an error state $f \in E$ (for an arbitrary initial state $s \in E$) is independent of the history. That is, the resilience condition is either simultaneously satisfied or simultaneously violated for all visits to $f$. 
However, if the second case holds, $\asched_E$ is not even average-resilient, what is a contradiction. Thus, we obtain:  

\begin{lemma}
	\label{lem-resilient}
	Let $(E,\A,\asched_E) \in \cE$, and let $s \in E$. Then the scheduler $\asched_E$ is \emph{resilient} when the initial state is changed to~$s$. Further, if  $\rsched$ is a resilient scheduler for $\cQ_E$ with the initial state $s$, then $\Avail^{\asched_E}_{{\cQ_E},s} \ \geq \ \Avail^{\rsched}_{{\cQ_E},s}$.
\end{lemma}

\vspace{-2em}
\subsubsection{Proof of Lemma~\ref{lem-claim-one}.}
Let $\rsched$ be a resilient scheduler for $\transformed{\cM}$. We show that there is another resilient scheduler $\rsched'$ satisfying the conditions of Lemma~\ref{lem-claim-one}. First, let us consider the end components $(F_1,\cB_1),\ldots,(F_m,\cB_m)$ where $\rsched$ stays. For every $(F_i,\cB_i)$, let $(E,\A,\asched_E) \in \cE$ be a triple with the maximal $\Avail(E)$ such that  $F_i \cap E \neq \varnothing$ (such a triple exists due to Lemma~\ref{lem-cover}). We say that $(E,\A,\asched_E)$ is \emph{associated} to $(F_i,\cB_i)$. Let $\Avail(F_i,\cB_i)$ be the \emph{conditional availability} w.r.t. scheduler $\rsched$ under the condition that an infinite path initiated in $\sinit$ stays in $(F_i,\cB_i)$.
Given a triple $(E,\A,\asched_E)\in \cE$, 
we use $\Avail(E)$ to denote the availability achieved by 
scheduler $\asched_E$ for $s$.
Note that $\Avail(E)$ is independent of~$s$.

\begin{lemma}
\label{lem-avegare-avail}
   $\Avail(F_i,\cB_i) \leq \Avail(E)$, where $(E,\A,\asched_E) \in \cE$ is the triple associated to $(F_i,\cB_i)$.   
\end{lemma}
Further, we say that $(F_i,\cB_i)$ is \emph{offending} if there is a finite $\rsched$-path $\fpath$ initiated in $\sinit$ ending in a state $s \in E$, where $(E,\A,\asched_E)$ is associated to $(F_i,\cB_i)$, such that $s \not\in \Repair$ and the availability achieved by the scheduler $\residual{\rsched}{\fpath}$ in $s$ is \emph{strictly larger} than $\Avail(E)$. Note that if no  $(F_i,\cB_i)$ is offending, we can choose $\cE'$ as the set of triples associated to 
$(F_1,\cB_1),\ldots,(F_m,\cB_m)$, and redefine the scheduler $\rsched$ into a resilient scheduler $\rsched'$ as follows: $\rsched'$ behaves exactly like $\rsched$ until a state $s$ of some $(E,\A,\asched_E) \in \cE'$ is visited. Then, $\rsched'$ switches to $\asched_{E}$ immediately. The scheduler $\rsched'$ is resilient because $s \not\in \Repair$ (a visit to a repair state is preceded by a visit to the associated fail state which also belongs to $E$) and hence we can apply Lemma~\ref{lem-res-switch}. Clearly, $\rsched'$ is consistent with every $\asched_{E}$ such that $(E,\A,\asched_E) \in \cE'$. It remains to show that the availability achieved by $\rsched'$ in $\sinit$ is not smaller than the one achieved by $\rsched$.  This follows immediately by observing that whenever $\rsched'$ makes a switch to $\asched_{E}$ after performing a finite $\rsched$-path initiated in $\sinit$ ending in $s \in E$, the availability achieved by the resilient scheduler $\residual{\rsched}{\pi}$ for the initial state~$s$ must be bounded by $\Avail(E)$, because otherwise some $(F_i,\cB_i)$ would be offending. So, the introduced ``switch'' can only increase the availability.

Now assume that $(F_m,\cB_m)$ is offending, and let $(E,\A,\asched_E)$ be the triple associated to $(F_m,\cB_m)$. We construct a resilient scheduler $\tilde{\rsched}$ which stays in $(F_1,\cB_1),\ldots,(F_{m-1},\cB_{m-1})$ and achieves availability not smaller than the one achieved by $\rsched$. This completes the proof of Lemma~\ref{lem-claim-one}, because we can then successively remove all offending pairs. 
Since $(F_m,\cB_m)$ is offending, there is a finite $\rsched$-path $\fpath$ initiated in $\sinit$ ending in a state $s \in E$ such that $s \not\in \Repair$ and the availability $A$ achieved by $\residual{\rsched}{\fpath}$ in $s$ is larger than $\Avail(E)$. Since $F_m \cap E \neq \varnothing$, there is a state $t \not\in \Repair$ such that $t \in F_m \cap E$. Note that $\asched_E$ is resilient for the initial state $t$, and almost all infinite paths initiated in $t$ visit the state~$s$ under the scheduler~$\asched_E$.
 
Now, we construct a resilient scheduler $\sched_s$ achieving availability at least $A$ in~$s$ such that all components where $\sched_s$ stays (for the initial state~$s$) are among $(F_1,\cB_1),\ldots,(F_{m-1},\cB_{m-1})$. Let $P_m$ be the probability that an infinite path initiated in $s$ stays in $(F_m,\cB_m)$  under the scheduler $\residual{\rsched}{\fpath}$. If $P_m = 0$, we put $\sched_s = \residual{\rsched}{\fpath}$. Now assume $P_m > 0$. We cannot have $P_m = 1$, because then $A$ is bounded by $\Avail(E)$ (see Lemma~\ref{lem-avegare-avail}). Let $B$ be the conditional  availability achieved in $s$ by $\residual{\rsched}{\fpath}$ under the condition that an infinite path initiated in $s$ stays in $(F_1,\cB_1),\dots,(F_{m-1},\cB_{m-1})$. Since $A \leq (1{-}P_m) \cdot B + P_m \cdot \Avail(E)$ and $A > \Avail(E)$, we obtain $B > A$. 
For every $\varepsilon >0$, let $\Pi^\varepsilon$ be the set of all finite ($\residual{\rsched}{\fpath}$)-paths $\pi'$ initiated in $s$ and ending in $t$ such that the probability of all infinite paths initiated in $t$ staying in $(F_m,\cB_m)$ under the scheduler $\residual{\rsched}{(\fpath;\fpath')}$ is at least $1 {-}\varepsilon$. 
Note that each $(\residual{\rsched}{\fpath})$-path initiated in $s$ and staying in  $(F_m,\cB_m)$ is included in  $(\residual{\rsched}{\fpath})$-paths starting with a prefix of $\Pi^\varepsilon$. 
Hence, a smart redirection of the strategy after passing via $\Pi^\varepsilon$ can avoid staying in $(F_m,\cB_m)$.
We use $P_m^\varepsilon$ to denote the probability (under the scheduler $\residual{\rsched}{\fpath}$) of all infinite paths initiated in $s$ starting with a prefix of $\Pi^\varepsilon$, and $B^\varepsilon$ to denote the conditional availability achieved in $s$ by $\residual{\rsched}{\fpath}$ under the condition that an infinite path initiated in $s$ does \emph{not} start with a prefix of $\Pi^\varepsilon$. Since  $\lim_{\varepsilon \rightarrow 0} P_m^\varepsilon = P_m$ and $\lim_{\varepsilon \rightarrow 0} B^\varepsilon = B$, we can fix a sufficiently small $\delta > 0$ where
\begin{itemize}
	\item[I.] $\delta\cdot M + (1{-}\delta)\cdot \Avail(E) < A$, where $M$ is the maximal payoff assigned to a state of $\transformed{\cM}$.
	\item[II.] conditional bound $B^\delta > A$.
\end{itemize}
The scheduler $\sched_s$ is defined in the following way, where $\Sigma$ denotes the set of all finite paths $\varrho$ initiated in $t$ and ending in $s$, such that the state $s$ is visited by $\varrho$ only once: 
\[
   \sched_s(\pi') = \begin{cases}
                         \sched_s(\pi'') & \mbox{if $\pi' = \hat{\pi};\varrho;\pi''$ where $\hat{\pi} \in \Pi^\delta$ and $\varrho \in \Sigma$},\\
                         \asched_E(\pi'') & \mbox{if $\pi' = \hat{\pi};\pi''$ where $\hat{\pi} \in \Pi^\delta$ and no prefix of $\pi''$ is in $\Sigma$},\\
                         (\residual{\rsched}{\fpath})(\pi') & \mbox{otherwise.}
                    \end{cases}
\]
Intuitively, $\sched_s$ simulates $\residual{\rsched}{\fpath}$ unless a path of $\Pi^\delta$ is produced, in which case $\sched_s$ temporarily ``switches'' to $\asched_E$ until $s$ is revisited and the simulation of $\residual{\rsched}{\fpath}$ is restarted. It is easy to verify that $\sched_s$ is a resilient scheduler achieving availability equal to $B^\delta > A$ staying in end components $(F_i,\cB_i)$ with  $i < m$. 

Now we can easily construct the scheduler $\tilde{\rsched}$. Let $\Xi^\delta$ be the set of all finite paths $\pi$ initiated in $\sinit$ and ending in $t$ where the probability of all infinite paths initiated in $t$ staying in $(F_{m},\cB_{m})$ is at least $1{-}\delta$. The scheduler $\tilde{\rsched}$ behaves as $\rsched$ unless a path of $\Xi^\delta$ is produced, in which case $\tilde{\rsched}$ temporarily switches to $\asched_E$ until the state $s$ is reached, and then it permanently switches to $\sched_s$. The availability achieved by $\tilde{\rsched}$ in $\sinit$ can be only larger that the availability achieved by $\rsched$ due to Conditions~I and~II above.

\subsubsection{Proof of Lemma~\ref{lem-claim-two}.}
Let $\cE$ denote the set of triples computed by Algorithm~\ref{alg-comp-compute}.
Due to Lemma~\ref{lem-claim-one}, we can concentrate on schedulers those paths almost 
surely reach subsets $\cE'\subseteq\cE$ and are consistent with the schedulers in $\cE'$.
Observe that the transient prefix of each path then has no effect on the long-run availability 
of the path and just influences the reachability probability distribution on $\cE$. 
The resulting availability then is a convex combination of availabilities of the triples in $\cE$.
Thus, the aim is to find a \emph{resilient} scheduler that maximizes this convex combination.
We do so by constructing an MDP $\cN$ where the resilient MR-scheduler $\rsched_{\cN}$ with optimal reachability reward induces optimal resilient scheduler in $\transformed{\cM}$.
We show that $\rsched_{\cN}$ can be obtained from a slightly modified 
linear program of \cite{Kallenberg,Puterman:book}.

Let $\cN=(S_\cN,\Act_\cN,P_\cN,\sinit)$ be an MDP over the state space
$$
S_{\cN} \ \ = \ \ 
\transformed{S} \cup 
\bigl\{ \, \goal_E \, : \, (E,\A,\asched_E) \in \cE \, \bigr\} 
\cup \{\goal\}
$$
and the action space $\Act_{\cN}={\transformed{Act}} \cup \{\tau\}$, 
where $\tau$ is a fresh action symbol.
The transition probabilities $P_\cN$ are defined as for $\transformed{\cM}$, 
but with additional $\tau$-transitions for each $(E,\A,\asched_E) \in \cE$:
\begin{itemize}
	\item from each state $\transformed{s} \in E \cap \transformed{\Operational}$ to $\goal_E$,
	i.e., $P_{\cN}(\transformed{s},\tau,\goal_E)=1$,
	\item from $\goal_E$ to $\goal$, i.e., $P_{\cN}(\goal_E,\tau,\goal)=1$, and
	\item from $\goal$ to $\goal$, i.e., $P_{\cN}(\goal,\tau,\goal)=1$.
\end{itemize}
The reward function in $\cN$ is given by 
$\rew(\goal_E)=\Avail(E)$ for each $\goal_E \in S_\cN$
and
$\rew({s})=0$ for all the remaining states ${s}\in \transformed{S}\cup \{\goal\}$.
Given a scheduler $\sched$, the random variable $TR$ assigns to an infinite $\sched$-path 
$\infpath = s_0 \, \act_0 \, s_1 \, \act_1 \, s_2 \, \act_2 \ldots$
the total accumulated reward $TR(\infpath) = \sum_{i=0}^{\infty} \rew(s_i).$ 
The expected total accumulated reward from a state $s \in S_{\cN}$ 
is denoted by $\Exp^{\sched}_{\cN,s}[TR]$.

\begin{lemma}
	\label{lem:MtoN}
	Let $\rsched'$ be a resilient scheduler for $\transformed{\cM}$ such that $\rsched'$-paths
	from $\sinit$ almost surely reach a subset $\cE'\subseteq\cE$ and is consistent with the schedulers in $\cE'$.
	Then, there is a resilient scheduler $\rsched$ for $\cN$ where the
	$\rsched$-paths from $\sinit$ almost surely reach $\goal$ and	
	$$\Avail^{\rsched'}_{\transformed{\cM},\sinit} = \E^{\rsched}_{\cN,\sinit}[TR] .$$

\end{lemma}

From $\rsched'$ we can easily construct an equivalent scheduler $\rsched$ by redefining $\rsched'$ to almost surely 
perform $\tau$ actions in $E \cap \transformed{\Operational}$ for $(E,\A,\asched_E) \in \cE'$.
From Lemma~\ref{lem-claim-one} and Lemma~\ref{lem:MtoN} it follows that if there is no resilient scheduler for $\cN$ %
there is no resilient scheduler for $\transformed{\cM}$.
Let $\rsched_{\cN}$ be the resilient scheduler that acquires the supremum of the expected total accumulated rewards from $\sinit$ among all resilient schedulers for $\cN$ that reach $\goal$ almost surely from $\sinit$.
As we shall see bellow, we can safely assume that $\rsched_{\cN}$ %
is an MR-scheduler. 
The technical details for proving the following lemma can be found in %
Appendix~\ref{app:claim-two}.

\begin{lemma}
	\label{lem:NtoM}
	Let $\rsched_{\cN}$ be an MR-scheduler that acquires maximal $\E^{\rsched'}_{\cN,\sinit}[TR]$ within resilient schedulers $\rsched'$ for $\cN$ such that almost all $\rsched'$-paths reach the $\goal$. Let $\cE'$ be the set of all $(E,\A,\asched_E) \in \cE$ such that $\goal_E$ is visited from $\sinit$ with positive probability under $\rsched_{\cN}$, and let
	$\sched_e(s) = \rsched_{\cN}(s)$ for each $s \in F$ where $F = \transformed{S} \setminus \bigcup_{(E,\A,\sched_E) \in \cE'} E$.
	Moreover, let  
	$\rsched$ be the unique scheduler consistent with $\sched_e$ and $\asched_E$ for each $(E,\A,\asched_E) \in \cE'$.
	It holds that
	$$\Avail^{\rsched}_{\transformed{\cM},\sinit} = \E^{\rsched_\cN}_{\cN,\sinit}[TR].$$
\end{lemma}

\noindent
Note that the scheduler $\rsched$ of Lemma~\ref{lem:NtoM} simulates the scheduler $\rsched_{\cN}$ only until a state of $\cE'$ is visited (not until $\rsched_{\cN}$ visits a $\goal_E$ state). This is the main subtlety hidden in Lemma~\ref{lem:NtoM}.

\paragraph*{A resiliency linear program.}
To obtain $\rsched_{\cN}$, let us consider the following linear program clearly constructible in
polynomial time in $|\cN|$ (and thus also in $|\transformed{\cM}|$). 
Intuitively, the variables $y_{t,\alpha}$ stand for the expected number of times an action 
$\alpha\in \Act_{\cN}$ is taken from state $t \in S_{\cN}$. 
We set $y_{t}= \sum_{\alpha\in \Act_{\cN}(t)} y_{t,\alpha}$ and define
\begin{enumerate}
	\item [(1)]
	flow equation: for all states $s \in S_{\cN} \setminus \{\goal\}$
	$$
	y_{s} \ \ = \ \ 
	\delta(s,\sinit) \ + \ \sum\nolimits_{t\in S_{\cN}} \  
	\sum\nolimits_{\alpha\in \Act_{\cN}(t)}
	y_{t,\alpha} \cdot 
	P_{\cN}(t,\alpha,s)
	$$
	where $\delta(s,\sinit)$ is $1$ if $s=\sinit$, and $0$ otherwise.
	\item [(2)] non-negativeness: \ 
	$y_{s,\alpha}\geqslant 0$ for all state-action pairs 
	$(s,\alpha)$. 
	\item [(3)] flow equation for the goal state: \ 
	$y_{\goal}\geqslant 1$.
	\item [(4)]
	resiliency constraint: for all $e\in \Error$ 
	$$
	\sum\nolimits_{s\in \Operational_e} y_{s} 
	\ \ \geqslant \ \ 
	\wp \cdot y_e
	$$
\end{enumerate}

The next lemma is proven by the methods of \cite{Kallenberg,Puterman:book} (the only difference distinguishing our case is Constraint~(4), which is easy to handle). 
\begin{lemma}
	\label{lem:reach_correctness}
	Each feasible solution $(z^*_{s,\alpha})_
	{s\in S_{\cN},\alpha \in \Act_{\cN}(s)}$ of the linear program 
	(1)-(4) under the objective to maximize
	$\sum_{(E,\A,\asched_E) \in \cE} y_{\goal_E} \cdot \Avail(E)$,
	induces an MR-scheduler $\rsched_\cN$ that is resilient in $\cN$
	and can be computed in time polynomial in $|\cN|$.
	If there is no such solution, there is no resilient scheduler in $\cN$.
	
	Conversely, let $\rsched$ be a resilient scheduler such that $\rsched$-paths 
	almost surely reach $\goal$ and the expected number of actions executed 
	before reaching $\goal$ is finite. 
	Let $z_{s,\alpha}$ denote the expected number of times an 
	action $\alpha \in \Act_{\cN}$ is taken in a state $s \in S_\cN$ using $\rsched$. 
	Then, values $z_{s,\alpha}=y_{t,\alpha}$ form a solution of the above linear constraints (1)-(4).
\end{lemma}

According to the second part of Lemma~\ref{lem:reach_correctness}, the scheduler $\rsched_\cN$ achieves the optimal total accumulated reward among all resilient schedulers where the expected number of transitions executed before reaching $\goal$ is finite. The next lemma shows that $\rsched_\cN$ achieves the optimal total accumulated reward among \emph{all} resilient schedulers, which completes the proof of Lemma~\ref{lem-claim-two}.

\begin{lemma}
	\label{lem:completeness}

	$\E^{\rsched_\cN}_{\cN,\sinit}[TR]\geq\ L$ with $L$ being the supremum
	over all $E^{\rsched}_{\cN,\sinit}[TR]$ ranging over resilient schedulers $\rsched$ 
	in $\cN$ those paths almost surely reach $\goal$.
\end{lemma}

\begin{proof}
	First, note that $\E^{\sched}_{\cN,\sinit}[TR]$ for an HR-scheduler $\sched$ can be approximated 
	up to an arbitrary small error using a sequence of schedulers $\rsched_i$:
	For each $i \in \Nat$ we define the scheduler $\rsched_i$ by acting as $\sched$ until the $i$-th step 
	and then continuing as $\rsched_{\cN}$.
	The expected number of executed actions before reaching the $\goal$ state is finite for all $\rsched_i$.
	Clearly, $|\E^{\sched}_{\cN,\sinit}[TR] - \E^{\rsched_i}_{\cN,\sinit}[TR]|$ gets arbitrarily small for increasing $i$.
	Towards a contradiction, assume that $L - \E^{\rsched_{\cN}}_{\cN,\sinit}[TR] > \delta > 0$.
	Then, there is a sequence of schedulers that approximate $L$ arbitrarily close and there is a
	scheduler $\rsched$ such that $\E^{\rsched}_{\cN,\sinit}[TR] = K$ with $|L-K| < \delta/2$.
	Moreover, there is sequence of schedulers $\rsched_i$ that approximate $K$ arbitrarily close and have a 
	finite expected number of executed actions before reaching $\goal$. 
	Hence, there is some $\rsched_i$ such that $\big|L - \E^{\rsched_i}_{\cN,\sinit}[TR]\big| < \delta$, which 
	is in contradiction with the optimality of $\rsched_{\cN}$ among all schedulers with a finite expected number 
	of actions executed before reaching $\goal$.
	\qed
\end{proof}

\subsection{A lower complexity bound}
\label{sec:complexity}
 
When the bound $R$ is encoded in binary, our algorithms become exponential. 
Using the PSPACE-hardness result for cost-bounded reachability
problems in acyclic MDPs by Haase and Kiefer \cite{HaaseKiefer15}, we show that the question whether there exists a resilient scheduler is PSPACE-hard, even for acyclic MDPs, when $R$ is encoded in binary.

\begin{lemma}
\label{lem-lower}
  If $R$ is encoded in binary, the problem to check the existence of a resilient scheduler
  and the decision variant of the resilience-availability problem are PSPACE-hard.
\end{lemma}
\begin{proof}
	In \cite{HaaseKiefer15}, the PSPACE-completeness of the following
	cost-problem has been proven: Given an acyclic MDP 
	$\cN = (S,\Act,P,\sinit)$ with a cost function
	and a cost bound $R$,
	the task is to check whether there is a scheduler
	$\sched$ for $\cN$ such that 
	$\Pr^{\sched}_{\cN,\sinit}(\Eventually^{\leqslant R} T) \geqslant \frac{1}{2}$.
	Here, $T$ denotes the set of trap states in $\cN$ and $\sinit \notin T$.
	
	We now provide a polynomial reduction from the cost-problem \`{a} la
	Haase and Kiefer \cite{HaaseKiefer15} to the problem to decide
	the existence of a resilient scheduler and the decision
	variant of the resilience-availability problem.
	
	Let $\cM$ be the MDP resulting from $\cN$ by defining $\Error=\{\sinit\}$
	and $\Operational = T$ and adding a fresh action symbol $\tau$ 
	and $\tau$-transitions from the states $t\in T$
	to $\sinit$. That is,
	$\cM$ has the same state space as $\cN$, the action set is
	$\Act_{\cM}=\Act \cup \{\tau\}$ and the $\cM$'s 
	transition probability function extends 
	$\cN$'s transition probability function by $P(t,\tau,\sinit)=1$
	and $P(t,\alpha,s)=0$ 
	for all states $t\in T$, $\alpha \in \Act_{\cM}$ and $s\in S$ with
	$(s,\alpha)\not= (\sinit,\tau)$.
	$\cM$'s cost function is the same as in $\cN$ for all states $s\in S$
	and $\cost(t)=0$ for all states $t\in T$.
	Obviously, each scheduler $\sched$ for $\cN$ with
	$\Pr^{\sched}_{\cN,\sinit}(\Eventually^{\leqslant R} T) \geqslant \frac{1}{2}$
	can be viewed as a memoryless resilient scheduler for $\cM$
	with respect to the probability threshold $\wp=\frac{1}{2}$ and cost bound $R$.
	Vice versa, given a resilient scheduler $\sched'$ for $\cM$,
	the decisions of $\sched'$ for the paths from $\sinit$ to a $T$-state yield a scheduler $\sched$ for $\cN$ with
	$\Pr^{\sched}_{\cN,\sinit}(\Eventually^{\leqslant R} T) \geqslant \frac{1}{2}$.
	
	For the decision problem of the resilience-availability problem, 
	we use the same reduction
	with availability threshold $\vartheta=0$ and the payoff function that
	assign 0 to all operational states.
\end{proof} %

\bibliographystyle{plain}
\bibliography{resilience}

\newpage
\appendix
\section{Proofs for Lemma~\ref{lem-claim-one}}

\textbf{Lemma~\ref{lem-avegare-avail}.}
\textit{$\Avail(F_i,\cB_i) \leq \Avail(E)$, where $(E,\A,\asched_E) \in \cE$ is the triple associated to $(F_i,\cB_i)$.}
\begin{proof}
	Let $A = \Avail(F_i,\cB_i)$. By contradiction we assume $A > \Avail(E)$. Then for an arbitrarily small $\varepsilon > 0$, there is a finite $\rsched$-path $\pi$ initiated in $\sinit$ ending in a state \mbox{$s \in F_i \setminus \Repair$} such that the probability (under the scheduler $\residual{\rsched}{\fpath}$) of all infinite paths initiated in~$s$ eventually staying in $(F_i,\cB_i)$ is at least $1 {-}\varepsilon$, and the availability achieved by  $\residual{\rsched}{\fpath}$ in $s$ is at least~$A$. For the initial state $s$, the scheduler $\residual{\rsched}{\fpath}$ is average resilient, but it can still choose \emph{leaving} transitions leading to the states outside $F_i$ (and then possibly outside $\cQ_E$) with positive probability, so it cannot be seen as a scheduler for $\cQ_E$.
	Now consider a scheduler $\sched^\varepsilon$ for the initial state $s$ which behaves like $\residual{\rsched}{\fpath}$ except that instead of executing a leaving transition,
	$\sched^\varepsilon$ selects some transition leading inside $F_i$ with the same probability (after that, $\sched^\varepsilon$ behaves arbitrarily, but it stays in $F_i$). Note that $\sched^\varepsilon$ is not necessarily average resilient. However, as $\varepsilon \rightarrow 0$, the availability achieved by $\sched^\varepsilon$ approaches the one achieved by $\residual{\rsched}{\fpath}$ (which is at least $A$), and $\fMP_e^{\sched^\varepsilon}$ approaches $0$ for all $e \in \Error$. Now we can apply the result of \cite{BBCFK-TwoViews14} which says that the set of achievable solutions for multiple mean-payoff objectives is closed under Pareto points. Note that $(F_i,\cB_i)$ can be seen as an MDP with initial state~$s$, and $\sched^\varepsilon$ are schedulers for this MDP. Hence, there must be an average resilient scheduler $\sched$ for $(F_i,\cB_i)$ achieving availability at least~$A$ in~$s$. Since $F_i \cap E \neq \varnothing$, there is $t \not\in \Repair$ such that $t \in F_i \cap E$. Consider a scheduler $\tsched$ for $\cQ_E$ which behaves like $\asched_E$ until $t$ is visited, and then it switches to~$\sched$. Then $\tsched$ is average resilient for all states of $E$ and achieves availability larger than $\Avail(E)$, which contradicts Lemma~\ref{lem-resilient}.
	\qed
\end{proof}

\section{Proofs for Lemma~\ref{lem-claim-two}}
\label{app:claim-two}

\medskip
\noindent
\textbf{Lemma~\ref{lem:MtoN}.}
\textit{
		Let $\rsched'$ be a resilient scheduler for $\transformed{\cM}$, such that 
		it is consistent with the schedulers in $\cE'$
		and almost all $\rsched'$-paths
		reach a subset $\cE'\subseteq\cE$ from $\sinit$. 
		There is a resilient scheduler $\rsched$ for $\cN$ almost surely reaching $\goal$ from $\sinit$ with
		$$\Avail^{\rsched'}_{\transformed{\cM},\sinit} = \E^{\rsched}_{\cN,\sinit}[TR] .$$
}
\begin{proof}{(Sketch)}
	Let $\pi$ be a finite path where $\last(\pi) = s$. 
	We set $\rsched(\pi)(\tau) = 1$ if $s \in (E \cap \transformed{\Operational}) \cup \{\goal, \goal_E\}$ for some $(E,\A,\asched_E) \in \cE'$ 
	and
	$\rsched(\pi) = \rsched'(\pi)$ otherwise.
	Scheduler $\rsched$ is resilient, since $\rsched'$ is resilient and we switch from its behavior only in operational states what does not effect resilience up to the switch and since that no more error is reached.
	
	Note that the probability to reach some $(E,\A,\asched_E) \in \cE'$ from $\sinit$
	is the same for $\rsched'$ and $\rsched$, i.e., when $p_E$ and $q_E$ denote the probabilities
	of reaching $E$ for some $(E,\A,\asched_E)\in\cE'$ under $\rsched'$ and $\rsched$, respectively, 
	then $p_E=q_E$. 
	Thus, 
	\[
	\Avail^{\rsched'}_{\transformed{\cM},\sinit} = \sum_{(E,\A,\asched_E) \in \cE'} p_E\cdot \Avail(E) =
	\sum_{(E,\A,\asched_E) \in \cE'} q_E\cdot \Avail(E) = \E^{\rsched}_{\cN,\sinit}[TR].
	\]
	\qed
\end{proof}

To prove Lemma~\ref{lem:NtoM}, we need auxiliary Lemma~\ref{lem:same_value}.
Let $\rsched_{\cN}$ be the scheduler that acquires maximal $\E^{\rsched'}_{\cN,\sinit}[TR]$ within resilient MR-schedulers for $\cN$ such that almost all $\rsched'$-paths reach the $\goal$.
Intuitively, it states that if the $\tau$ action is chosen by $\rsched_{\cN}$ with positive probability to $\goal_E$ for some $(E,\A,\asched_E) \in \cE$ then the expected total reward of each state in $E \cap S$ equals to $\rew(\goal_E)$.
The lemma follows from properties of the expected total reward and since $(E,\A,\asched_E)$ is strongly connected.

\begin{lemma}
	\label{lem:same_value}
	Let $\rsched_{\cN}$ be an MR-scheduler that acquires maximal $\E^{\rsched'}_{\cN,\sinit}[TR]$ within resilient schedulers $\rsched'$ for $\cN$ such that almost all $\rsched'$-paths reach the $\goal$.
	Let $(E,\A,\asched_E) \in \cE$ and $s \in E$ such that $\rsched_\cN(s)(\tau)>0$. Then for all $s' \in E \cap S$
	$$\E^{\rsched_\cN}_{\cN,s'}[TR] = \rew(\goal_E).$$
\end{lemma}

\begin{proof}{(Sketch)}
	Observe that from the definition of the total accumulated reward it follows that if the MR-scheduler $\rsched_\cN$ is changed in a subset of states $S' \subseteq S_\cN$ to $\rsched$ such that $\E^{\rsched}_{\cN,s}[TR] > \E^{\rsched_\cN}_{\cN,s}[TR]$ for all $s \in S'$, then for each $s' \in S_\cN$ it holds that $\E^{\rsched}_{\cN,s'}[TR] \geq \E^{\rsched_\cN}_{\cN,s'}[TR]$.
	
	To prove the lemma, let us assume contrary and derive a contradiction.
	First, assume that there are states $s,s' \in E \cap S$ such that $\E^{\rsched_\cN}_{\cN,s}[TR] \neq \E^{\rsched_\cN}_{\cN,s'}[TR]$.
	Let $s'' \in E \cap S$ be the state with maximal $\E^{\rsched_\cN}_{\cN,s''}[TR]$.
	If $s'' \not\in \transformed{\Error}$ then we set $\rsched(s'') = \rsched_\cN(s'')$ and $\rsched(s) = \asched_E(s)$ for all other states $s \in E$.
	If $s'' \in \transformed{\Error}$ then we set $\rsched(s'') = \rsched_\cN(s'')$, $\rsched(\<s'',s,r\>) = \rsched_\cN(\<s'',s,r\>)$ for $\<s'',s,r\> \in \Repair $, and $\rsched(s) = \asched_E(s)$ for all other states in $E$.
	In both cases, $\rsched$ is a resilient MR-scheduler, since $\rsched_\cN$ and $\asched_E$ are resilient MR-schedulers. %
	Moreover, we strictly improved the total accumulated reward for some subset of states $ S' \subseteq (E \cap S)\setminus \{s''\}$, i.e., $\E^{\rsched}_{\cN,s}[TR] > \E^{\rsched_\cN}_{\cN,s}[TR]$ for all $s \in S'$.
	This is contradiction with the optimality of $\rsched_\cN$.
	Thus, the expected total accumulated reward is the same in all states of $E \cap S$.
	
	Let $s \in E \cap S$ and $\rsched_\cN(s)(\tau)>0$.
	Now, assume that $\E^{\rsched_\cN}_{\cN,s}[TR] < \rew(\goal_E)$.
	Then setting $\rsched_\cN(s)(\tau) = 1$ will result in 
	$\E^{\rsched_\cN}_{\cN,s'}[TR] = \rew(\goal_E)$, what is a contradiction with the optimality of $\rsched_\cN$.
	
	Finally, assume that $\E^{\rsched_\cN}_{\cN,s}[TR] > \rew(\goal_E)$.
	We change $\rsched_\cN$ by adding probability $\rsched_\cN(s)(\tau)$ proportionally to all other actions of $s$ and setting $\rsched_\cN(s)(\tau)$ to $0$.
	This strictly improves $\E^{\rsched_\cN}_{\cN,s}[TR]$ and is contradiction with the optimality of $\rsched_\cN$.
	\qed
\end{proof}

\medskip
\noindent
\textbf{Lemma~\ref{lem:NtoM}.}
\textit{	Let $\rsched_{\cN}$ be an MR-scheduler that acquires maximal $\E^{\rsched'}_{\cN,\sinit}[TR]$ within resilient schedulers $\rsched'$ for $\cN$ such that almost all $\rsched'$-paths reach the $\goal$. Let $\cE'$ be the set of all $(E,\A,\asched_E) \in \cE$ such that $\goal_E$ is visited from $\sinit$ with positive probability under $\rsched_{\cN}$, and let
	$\sched_F(s) = \rsched_{\cN}(s)$ for each $s \in F$ where $F = \transformed{S} \setminus \bigcup_{(E,\A,\sched_E) \in \cE'} E$.
	Moreover, let  
	$\rsched$ be the unique scheduler consistent with $\sched_F$ and $\asched_E$ for each $(E,\A,\asched_E) \in \cE'$.
	It holds that
	$$\Avail^{\rsched}_{\transformed{\cM},\sinit} = \E^{\rsched_\cN}_{\cN,\sinit}[TR].$$}
\medskip

\begin{proof}{(Sketch)}
	Let $(E,\A,\asched_E) \in \cE'$.
	Observe that, $E$ can be reached from $\transformed{S} \setminus E$ only through states in $E \cap S$.
	From Lemma~\ref{lem:same_value}, for each state $s \in E \cap S$ it holds that $\E^{\rsched_\cN}_{\cN,s}[TR] = \rew(\goal_E)$.
	This implies $\E^{\rsched_\cN}_{\cN,\sinit}[TR] = \Avail^{\rsched}_{\transformed{\cM},\sinit}$.
	
	Let $f \in \transformed{\Error}$.
	Note that, if $\Operational_f \cap E \neq \varnothing$ then $f \in E$.
	Then, since $\sched_F$ and $\asched_{E'}$ are resilient for each $(E',\A',\asched_{E'}) \in \cE'$ the condition~\eqref{res-sched} is satisfied for $\rsched$.	
	Condition~\eqref{res-rep} is satisfied for each $f \in \transformed{\Error} \cap E$, since $\asched_E$ is resilient and $(E,\A,\asched_E)$ is strongly connected.

	Let $f \in \transformed{\Error} \cap F$ and $\Pr\nolimits^{\rsched}_{\transformed{\cM},\sinit}(\Eventually f) > 0$.
	Since, $\sched_F$ was created from resilient $\rsched_{\cN}$ if the corresponding operational and repair states are all in $F$ the condition~\eqref{res-rep} holds for $f$.
	Assume that there is $\rsched$-path $s_0\alpha_0s_1\alpha_1\ldots$, where $s_0=f $ and there is $n>0$ such that $s_n \in E$ for some $(E,\A,\asched_E) \in \cE'$ and $ s_i \in \Repair$ for each $i \leq n$.
	Observe that $E \cap \transformed{\Operational} \neq \emptyset$ since $(E,\A,\asched_E) \in \cE'$ and the $\tau$ actions to $\goal_E$ are available only from $E \cap \transformed{\Operational}$ states.
	Moreover, $(E,\A,\asched_E)$ is strongly connected from the definition, 
	thus there is probability $1$ to reach $E \cap \transformed{\Operational}$ from $s_n$.
	\qed
\end{proof}

 \section{Lower complexity bound}

\textbf{Note on the proof of Lemma~\ref{lem-lower}.}
		In the approach of \cite{HaaseKiefer15} the cost function is integrated
		in the transition probability function. Formally, \cite{HaaseKiefer15}
		deals with a function
		$\delta\colon S \times \Act \to D(S \times \Nat)$
		where $D(S \times \Nat)$ denotes the set of probability distributions
		for $S \times \Nat$ with finite support.
		Rephrased for our notations, this means that we deal with the state
		space $S'=S \cup X$ where
		$X=\{(s,\alpha,s',k): \delta(s,\alpha)(s',k)>0\}$,
		the action set $\Act' = \Act \cup \{\iota\}$ where $\Act'(s)=\Act(s)$
		for $s\in S$ and $\Act'(s,\alpha,s',k)=\{\iota\}$ and
		the transition probabilities
		$P(s,\alpha,(s,\alpha,s',k))=\delta(s,\alpha)(s',k)$
		and
		$P((s,\alpha,s',k),\iota,s')=1$ and
		$P(\cdot)=0$ in all remaining cases.
		The cost function assigns value 0 to all states $s\in S$ and
		value $k$ to the states $(s,\alpha,s',k)\in X$.

\end{document}